\documentclass[10pt]{amsart}

\usepackage[foot]{amsaddr}
\usepackage{mathrsfs,amssymb,amsmath}
\usepackage{verbatim}
\usepackage{hyperref}
\usepackage{graphicx}
\usepackage{enumerate}
\usepackage{vmargin}
\usepackage{caption}
\usepackage{subcaption}
\usepackage{layout}
\usepackage{bigints}
\usepackage{booktabs}

\renewcommand{\Re}{\operatorname{Re}}

\newcommand{\R}{\mathbb{R}}

\newcommand{\E}{\mathbb{E}}
\newcommand{\cal}{\mathcal}

\renewcommand{\P}{\mathbb{P}}
\newcommand{\var}{\operatorname{Var}}

\newcommand{\prob}{\mathbb{P}}
\renewcommand{\c}{\mathrm{C}}

\allowdisplaybreaks

\def\br#1{\left(#1\right)}
\def\brb#1{\left[#1\right]}

\newtheorem{theorem}{Theorem}[section]
\newtheorem{proposition}[theorem]{Proposition}
\newtheorem{lemma}[theorem]{Lemma}
\newtheorem{corollary}[theorem]{Corollary}

\newtheorem{conjecture}[theorem]{Conjecture}

\theoremstyle{definition}

\numberwithin{equation}{section}

\begin{document}

%**************************************************************V
%\layout
%**************************************************************

\title[Two point function for critical points of a random plane wave]
{Two point function for critical points\\ of a random plane wave}
\author{Dmitry Beliaev}
\address{Mathematical Institute, University of Oxford}
\email{belyaev@maths.ox.ac.uk}
\author{Valentina Cammarota}
\email{valentina.cammarota@uniroma1.it}
\address{Department of Mathematics, King's College London and Sapienza Universit\`a di Roma}
\author{Igor Wigman}
\address{Department of Mathematics, King's College London }
\email{igor.wigman@kcl.ac.uk}


\date{\today}

%\thanks{ The research leading to these results has received funding from the European Research Council under the European Union's Seventh Framework Programme (FP7/2007-2013), ERC grant agreement n$^{\text{o}}$ 335141 (V. C. and I.W.) and from Engineering \& Physical Sciences Research Council (EPSRC) Fellowship EP/M002896/1 (D. B)}

\begin{abstract}
Random plane wave is conjectured to be a universal model for high-energy eigenfunctions of the Laplace operator on generic compact Riemanian manifolds. This is known to be true on average. In the present paper  we discuss one of important geometric observable: critical points. We first compute one-point function for the critical point process, in particular we compute the expected number of critical points inside any open set. After that we compute the short-range asymptotic behaviour of the two-point function. This gives an unexpected result that the second factorial moment of the number of critical points in a small disc scales as the fourth power of the radius.
\end{abstract}

\maketitle

\section{Introduction and main results}

\subsection{Random Gaussian functions}

Studying the Laplace eigenfunctions and their geometry is a classical subject going back to at least XIX century. It is most important to understand the eigenfunctions behaviour in the high energy limit. For a given domain, this is a difficult question, and we only have limited information about it other than in the few cases where the eigenfunctions is explicitly given.

For a generic chaotic domain (i.e. where the billiard dynamics is chaotic) it was conjectured by Berry \cite{Berry77} that the high energy functions behave like a random superposition of monochromatic plane waves propagating in different (random) directions, usually referred to as the {\em random plane wave}, rigorously defined below. As the comparison between these two is lacking mathematical rigour, one may understand this comparison in different ways.

Berry's conjecture seems to be out of reach by modern analytic techniques; a similar statement for a random linear combination of eigenfunctions with close eigenvalues could be proved though. Namely, for a compact Riemanian manifold $\mathcal{M}$ we can consider an orthonormal basis of eigenfunctions $\phi_i$ satisfying $\Delta \phi_i+t_i^2\phi_i=0$ with $t_0\le t_1\le \dots$, and define the {\em band-limited functions}
$$
f_T=\sum_{T-\sqrt{T}\le t_i\le T} c_i\phi_i
$$
where $c_i$ are i.i.d. normal random variables. It is known \cite{Lax,Hormander,Zelditch09,CaHa16} that the local scaling limit of $f_T$ is the random plane wave.

The above conjectures and results show that the random plane wave is a {\em universal} object, and motivate their further study; here we are interested in their {\em geometry}. As usual, a Gaussian random field could be defined or constructed in two different ways. On one hand we may define it as a concrete random series, an on the other hand we may describe it as uniquely defined in terms of it covariance function via Kolmogorov's Theorem. As a concrete random series we define the random plane wave with energy $E=k^2$ to be
\begin{equation}
\label{eq:RPW definition}
\Psi(z)=\Psi(r ,\theta) = \Re \sum_{n=-\infty}^{\infty} a_n J_{|n|} (k\, r) e^{i n \theta}
\end{equation}
in polar coordinates,
where $J_n$ are Bessel functions and $a_n$ are independent complex Gaussian random variables with variance $2$. Since the Bessel functions decay exponentially fast as functions of the order $n$, the series \eqref{eq:RPW definition} is almost surely convergent, absolutely and uniformly on any compact set, and hence the sum is a real analytic function.

By the definition \eqref{eq:RPW definition}, $\Psi$ is a centred Gaussian random field,
therefore its law is prescribed by the covariance function
\begin{align*}
\psi(z,w):=\mathbb{E}[\Psi(z)\cdot \Psi(w)],
\end{align*}
for $z,w \in {\mathbb R}^2$. It is then easy to evaluate $\psi$ explicitly as
$$
\psi(z,w)=J_0(k|z-w|),
$$
where $J_{0}$ is the Bessel $J$ function of order $0$.
From this representation it follows that $\Psi$ is stationary (i.e. its law is translation invariant), and isotropic, (i.e. its law is invariant under rotations); by the standard abuse of notation we write $\psi(z,w)=\psi(z-w)$.

It also follows directly from \eqref{eq:RPW definition} that the function $\Psi$ is an a.s. solution of the Helmholtz equation
\begin{equation}
\label{eq:Helmholz}
(\Delta + k^2) \Psi = 0,
\end{equation}
i.e. $\Psi$ is an a.s. eigenfunction of $-\Delta$ with eigenvalue $k^2$; we are interested in the geometry of random (or ``typical") solutions  $\Psi$ of \eqref{eq:Helmholz}. The geometric properties considered below are related to the nodal lines (i.e. $\Psi^{-1}(0)$), nodal domains (i.e. connected components of the complement of the nodal set), as well as the level curves ($\Psi^{-1}(c)$), and excursion sets (connected components of $\{z:\: \Psi(z)>c\}$). The geometry of these sets is closely related to that of the set of {\em critical points} of $\Psi$. The critical points and values and their applications appear a lot (\cite{NS09,So12,NS15,G-W,Estrade} to mention a few) in the literature on nodal domains of random plane waves and, more generally, smooth Gaussian fields.

%This random function is of special interest since it is conjectured to be a good universal description of high-energy Laplace eigenfunctions in domains with chaotic dynamics \cite{berry}. \\

\subsection{Critical points}
There are several intriguing questions on the critical points of random fields. From our perspective, of the most important questions are the ones on the distribution of the critical points {\em number}, and the corresponding critical {\em values}.
This general question could be made more concrete in different ways, most basically, evaluating the expected number of critical points inside a given domain; the latter admits a precise answer in a more general scenario. In an analogous case of random spherical harmonics (that converges to $\Psi$ as a scaling limit), Cammarota, Marinucci and Wigman \cite{CaMaWi} evaluated the expected number of minima, maxima and saddles whose value falls into a given window, and also determined the order of magnitude of the corresponding variance for a ``generic" window, which does not include the total number of critical points. In a subsequent work Cammarota and Wigman \cite{CaWi} resolved this outstanding case by evaluating the variance of the total number of critical points to be of lower order as compared to the generic case.

\begin{figure}[h]
\begin{subfigure}{.33\textwidth}
  \centering
  \includegraphics[scale=0.4]{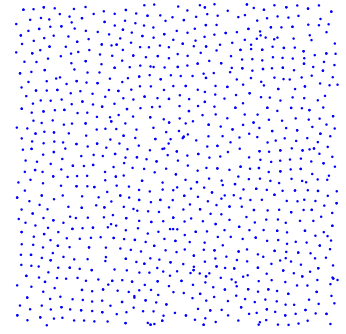}%{critpoints_color_expected_918_points_v2.pdf}
%  \caption{}
%  \label{}
 \end{subfigure}%
\begin{subfigure}{.33\textwidth}
  \centering
  \includegraphics[scale=0.5]{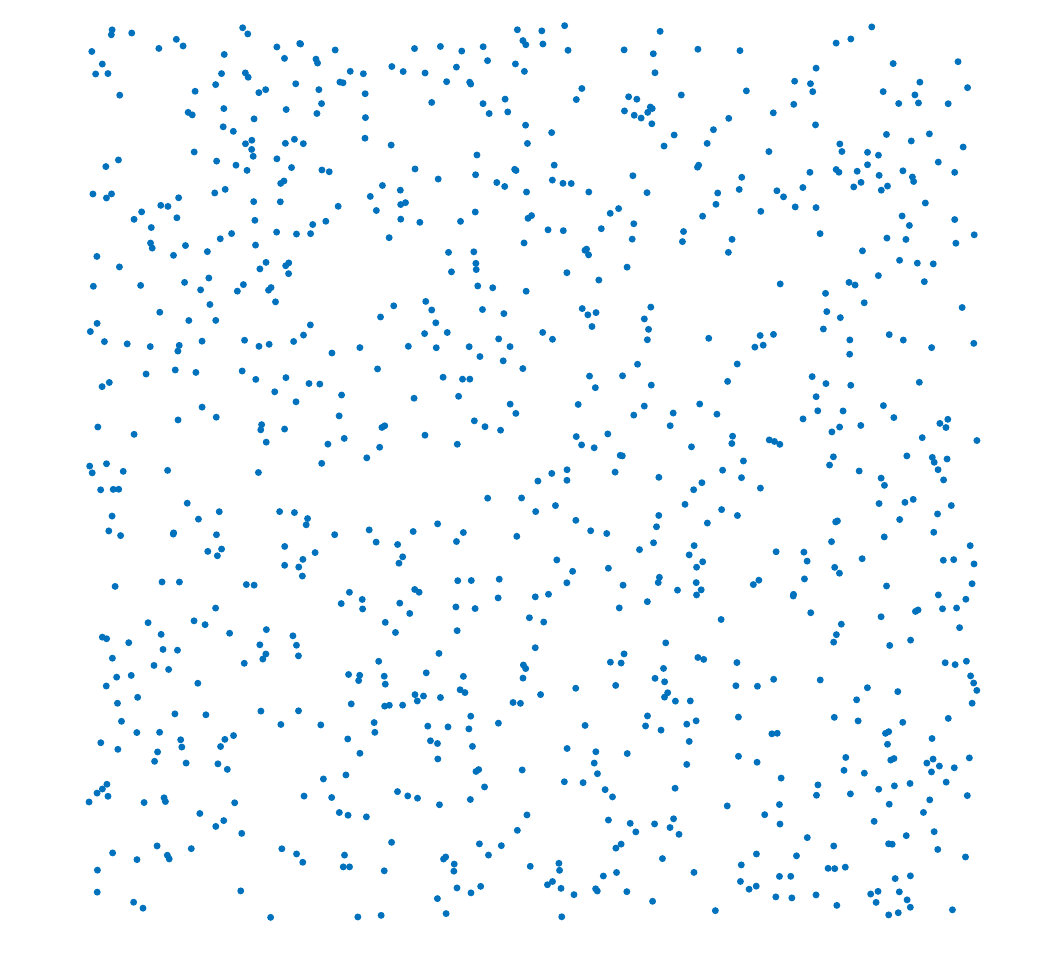}
% \caption{}
%  \label{}
\end{subfigure}
\begin{subfigure}{.33\textwidth}
  \centering
  \includegraphics[scale=0.5]{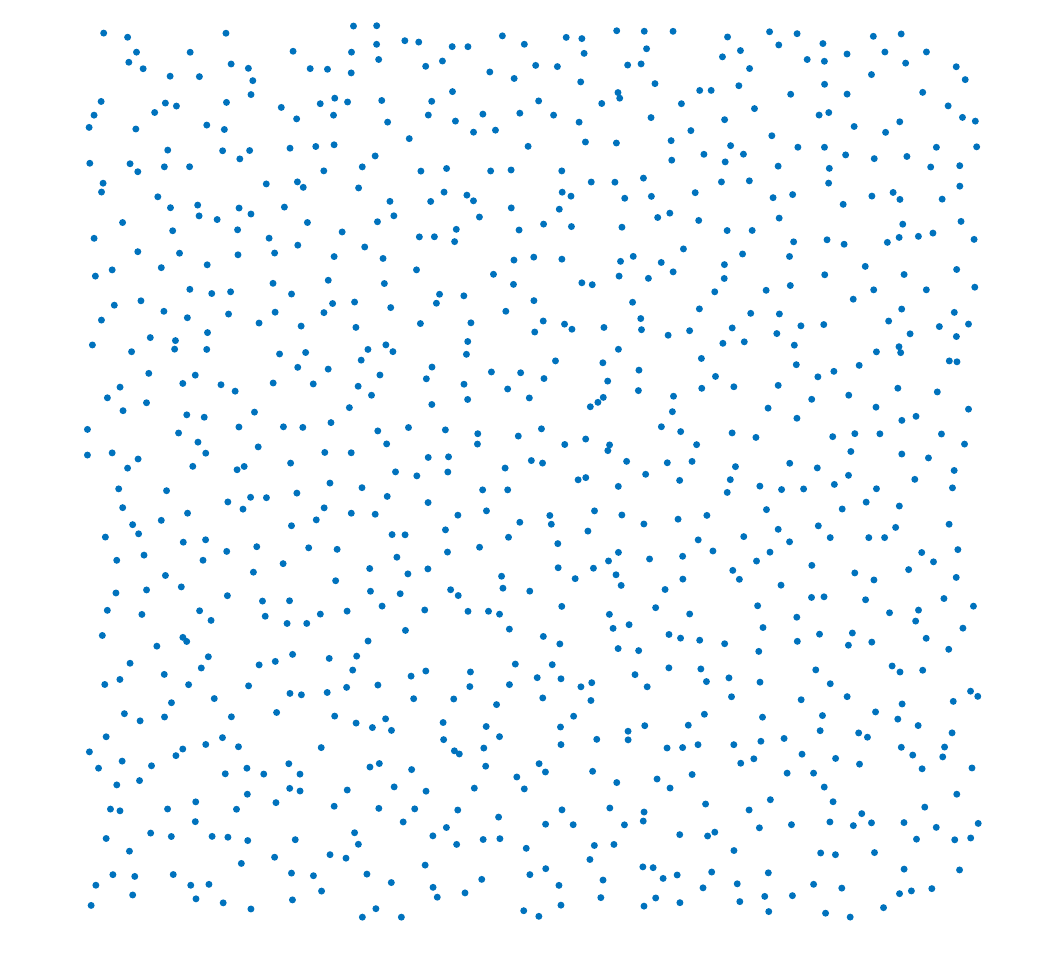}
% \caption{}
%  \label{}
\end{subfigure}
\caption{Left: critical points of a random plane wave. Center: The Poisson point process which has the same density. Right: a bulk part of the Ginibre ensemble with the same density. }
\label{fig:Point processes}
\end{figure}

It is important to understand the finer aspects of the {\em structure} of critical points. Upon looking at Figure \ref{fig:Point processes} (left),  it is evident that the structure of critical points is very ``rigid" or ``regular"; however it is not entirely clear how to formulate or quantify this statement with mathematical rigour. One can compare this to two other very well known translation invariant processes:
Figure \ref{fig:Point processes} (centre) one may observe the Poisson point process, and Figure \ref{fig:Point processes} (right) is the
corresponding picture for Ginibre point process; both are scaled to have the same intensity as the critical points in Figure
\ref{fig:Point processes} (left).

For all three point processes depicted in \ref{fig:Point processes} the number of points in a square of side-length $n$ is $c\cdot n^2$ where $c=1/2\sqrt{3}\pi$. This value of $c$ is the natural intensity of critical points (see Proposition \ref{expp}) of $\Psi$, whereas the other two point processes are so rescaled. The fluctuations of the total number of points in a square depend a lot on the point process. Though formally
stated for random spherical harmonics (which are only equivalent to $\Psi$ in the limit, under a natural scaling), it is likely that one
may deduce from \cite{CaWi} that the variance for the critical points scales like $n^2\log(n)$, whereas
for the Poisson point process it is asymptotic to $c\cdot n^2$ (with the same $c$ as above), and for the Ginibre ensemble it is of order $n$.

On the {\em local} scale, the probability that there is at least one point in a small disc or radius $\rho$ is the same for all three processes due to the translation invariance and our choice of normalization. The respective probabilities that there are exactly two points in a small disc are very different though. For the Poisson point process it is the probability that a Poisson random variable with intensity $c\pi \rho^2$ is equal to $2$. By the definition it is given by
$$
\prob(2\text{ points})=\frac{(c\pi \rho^2)^2\exp(-c\pi\rho^2)}{2}\approx \frac{c^2\pi^2\cdot \rho^4}{2}=\frac{1}{2^3 3}\cdot \rho^4,
$$
whereas for the Ginibre ensemble (which is a determinantal point process) this probability is of order $\rho^6$. That means that the points corresponding to the Ginibre ensemble repel each other, inducing on their visible regularity or rigidity.

Our principle result (Theorem \ref{sfm}) is the evaluation of the $2$nd factorial moment of the number of critical points of $\Psi$ in a radius $\rho$ disc, asymptotically for $\rho\rightarrow 0$. This suggests (Corollary \ref{prob012} and Conjecture \ref{uno})
that the probability of having precisely two critical points in the disc is
$$
\frac{1}{2^6 3\sqrt{3}}\rho^4+o(\rho^4),
$$
of the same order of magnitude (and leading constant smaller by the factor $8\sqrt{3}\approx 13.8$) as the probability of finding $2$-points in the same disc for the Poisson point process. This minor difference could not stand for the striking difference in the appearance of the two processes, highly regular for the critical points of $\Psi$. It is also worth noting, that despite the fact that critical points are more ``lattice like'' than the Ginibre ensemble, for the critical points clustering is significantly more likely (see Figure \ref{fig:Point processes blowup}).

\begin{figure}[h]
  \centering
  \includegraphics[scale=0.7]{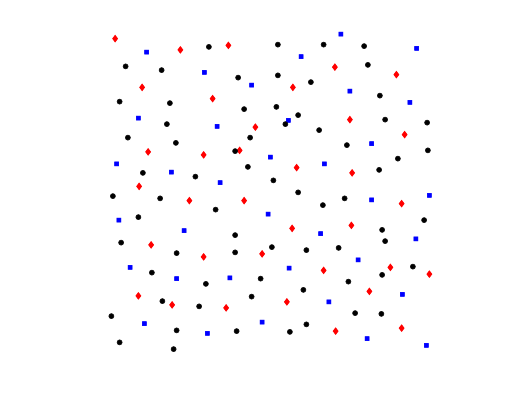}%
\caption{Central part of the critical points process from Figure  \ref{fig:Point processes}. Here we distinguish between different types of critical points: red diamonds are local minima, blue squares are local minima, and black discs are saddles.}
\label{fig:Point processes blowup}
\end{figure}

To formulate our main results we introduce the following notation for the number of critical points of a random plane wave $\Psi$ in a disc $\cal{B}(\rho)$ of radius $\rho>0$:
$$
{\cal N}^c_{\rho}=\# \{x \in \cal{B}(\rho): \nabla \Psi(x)=0\}.
$$
The numbers ${\cal N}^{saddle}_{\rho}$, ${\cal N}^{min}_{\rho}$, ${\cal N}^{max}_{\rho}$, and ${\cal N}^{e}_{\rho}$ of saddles, minima, maxima, and extrema respectively may also be defined.

Since the function $\Psi$ is translation invariant, the above random variables are independent of the center of the disc, so for simplicity, we may assume that it is centred at the origin. Another useful observation is that the random plane waves are scale invariant (that is, the law of $\Psi$ with arbitrary $k$ on $\cal{B}(1)$ is (up to homothety) equivalent to the law of $\Psi$ with $k=1$ on $\cal{B}(k)$); hence, with no loss of generality, we may assume that $k=1$, as we will for the rest of this manuscript.
The following principal results of this manuscript evaluate the expectation and the second factorial moment of ${\cal N}^c_{\rho}$ for small values of $\rho$. The first result, consistent to a similar statement from \cite[Proposition 1.1]{CaMaWi}, is for the expectations.

\begin{proposition} \label{expp}
For every $\rho >0$ we have
\begin{equation}
\label{eq: first moment}
\E[{\cal N}^c_{\rho}]= \frac{1}{2 \sqrt{3}} \rho^2
\end{equation}
and
\begin{equation}
4\E[{\cal N}^{min}_{\rho}]=4\E[{\cal N}^{max}_{\rho}]=2\E[{\cal N}^{saddle}_{\rho}]=2\E[{\cal N}^{e}_{\rho}]=\E[{\cal N}^{c}_{\rho}].
\end{equation}
\end{proposition}

Note that \eqref{eq: first moment} is not an asymptotic result, but rather a precise identity. More generally, same proof works on any open domain $\Omega$, i.e. the expected number of critical points lying in $\Omega$ is equal to $\frac{\mathrm{Area}(\Omega)}{2 \sqrt{3}\pi}$.
Evaluating the second moment is more involved, and we were unable to obtain a precise expression.
Instead, we show how it behaves asymptotically as  as the radius $\rho \to 0$.

\begin{theorem} \label{sfm}
As $\rho \to 0$, we have the following expansion for the number of critical points:
\begin{equation}
\label{eq:second moment}
\E[{\cal N}^c_{\rho} \; ({\cal N}^c_{\rho}-1)]
= \frac{1}{2^5 \, 3 \, \sqrt 3 } \; \rho^4+O(\rho^6).
\end{equation}
For ${\cal N}^{saddle}_{\rho}$, ${\cal N}^{min}_{\rho}$, ${\cal N}^{max}_{\rho}$, and ${\cal N}^{e}_{\rho}$ -- numbers of saddles, local minima,  local maxima, and local extrema in a ball of radius $\rho$ we have
\begin{align}
\label{eq:max-max moment}
&\E[{\cal N}^{max}_{\rho}({\cal N}^{max}_{\rho}-1)]=\E[{\cal N}^{min}_{\rho}({\cal N}^{min}_{\rho}-1)]=O(\rho^7\log(1/\rho)),
\\
\label{eq:extremum-extremum}
&\E[{\cal N}^{e}_{\rho}({\cal N}^{e}_{\rho}-1)]=O(\rho^7\log(1/\rho)),
\\
\label{eq:saddle-saddle moment}
&\E[{\cal N}^{saddle}_{\rho}({\cal N}^{saddle}_{\rho}-1)]=O(\rho^7\log(1/\rho)),
\\
\label{eq:max-min moment}
&\E[{\cal N}^{max}_{\rho}{\cal N}^{min}_{\rho}]=O(\rho^{12}),
\\
\label{eq:extremum-saddle moment}
&\E[{\cal N}^{e}_{\rho}{\cal N}^{saddle}_{\rho}]=\frac{1}{2^6 \, 3 \, \sqrt 3 } \; \rho^4+O(\rho^6).
\end{align}
\end{theorem}

There is no evidence that the estimates \eqref{eq:max-max moment}, \eqref{eq:saddle-saddle moment}, and  \eqref{eq:max-min moment} are sharp. 
In fact, it seems quite likely, that they are not, for \eqref{eq:max-min moment} particularly. Since the extremum-saddle covariance 
\eqref{eq:extremum-saddle moment} gives the main contribution to \eqref{eq:second moment}, the last formula \eqref{eq:extremum-saddle moment} is an asymptotic and as such gives a precise decay rate.
For integer-valued random variables it is more natural to consider factorial moments instead of the usual moments. The asymptotic behaviour of the variance, dominated by the expectation (and hence less useful), can be easily obtained by combining \eqref{eq: first moment} 
and \eqref{eq:second moment}
$$
\var\brb{{\cal N}^c_{\rho}}=\frac{1}{2 \sqrt{3}} \rho^2-\frac{8\sqrt{3}-1}{2^53\sqrt{3}}\rho^4+\dots.
$$
Since all our random variables ${\cal N}_\rho$ are integer valued, the first and second factorial moments yield the asymptotics for probabilities of the events ${\cal N}_\rho=1$ and ${\cal N}_\rho\ge 2$, as follows.

\begin{corollary} \label{prob012}
As $\rho \to 0$ we have the following asymptotic formulas for probabilities to have exactly one point:
\begin{equation}
\label{eq:prob(crit=1)}
\begin{aligned}
&\mathbb{P}({\cal N}^c_{\rho}=1)= \frac{1}{2 \sqrt{3}} \rho^2 + O(\rho^4),
\\
&\mathbb{P}({\cal N}^{min}_{\rho}=1)=\mathbb{P}({\cal N}^{max}_{\rho}=1)= \frac{1}{8 \sqrt{3}} \rho^2 + O(\rho^4),
\\
&\mathbb{P}({\cal N}^{e}_{\rho}=1)= \frac{1}{4 \sqrt{3}} \rho^2 + O(\rho^4),
\\
&\mathbb{P}({\cal N}^{saddle}_{\rho}=1)= \frac{1}{4 \sqrt{3}} \rho^2 + O(\rho^4).
\end{aligned}
\end{equation}
For probabilities to have at least two points we have
\begin{equation}
\label{eq:two point probabilities}
\begin{aligned}
&\P({\cal N}^c_{\rho}\ge 2) = O(\rho^4),
\\
&\P({\cal N}^{min}_{\rho}\ge 2) =
\P({\cal N}^{max}_{\rho}\ge 2) =O(\rho^7 \log(1/\rho)),
\\
&\P({\cal N}^{saddle}_{\rho}\ge 2)  =O(\rho^7 \log(1/\rho)),
\\
&\P({\cal N}^{min}_{\rho}\ge 1, {\cal N}^{max}_{\rho}\ge 1)=O(\rho^{12})
\end{aligned}
\end{equation}
\label{eq:three point probability}
Finally, for the probability to have three points we have
\begin{equation}
\P({\cal N}^c_{\rho}\ge 3)=O(\rho^7 \log(1/\rho)).
\end{equation}

\end{corollary}

\begin{proof} The proof is straightforward. For the sake of notational convenience we write $N={\cal N}^c_{\rho}$. The expectation and the second factorial moment could be written as series
\begin{equation*}
%\label{eq:sum k*prob(N=k)=exp}
\begin{aligned}
\E\brb{N}&=\P(N=1)+2\P(N=2)+3\P(N=3)+\dots
\\
\E\brb{N(N-1)}&=2\P(N=2)+6\P(N=3)+12\P(N=4)+\dots
\end{aligned}
\end{equation*}
Comparing the coefficients in front of $\P(N=k)$ we see that if the second moment is o-small of the expectation, then the expectation is dominated by $\P(N=1)$. In this case $\P(N=1)$ has the same leading term as the expectation and the error term is of the same order as the second moment. This, combined with the results of Proposition \ref{expp}, proves formulas \eqref{eq:prob(crit=1)}.
The estimates \eqref{eq:two point probabilities} are obtained by applying Markov inequality to the results of Theorem \ref{sfm}. Finally, to prove \eqref{eq:three point probability}, we notice that the event ${\cal N}^c\ge 3$ is majorised by ${\cal N}^{e}\ge 2$ or ${\cal N}^{saddle}\ge 2$.
\end{proof}

Note that the third factorial moment should be dominated by the event $N=3$, which, by \eqref{eq:three point probability}, is $O(\rho^7\log(1/\rho))$. This gives a strong evidence that
$$
\E\brb{N(N-1)(N-2)}=o(\rho^4).
$$
Assuming that this is indeed true, we can repeat the argument above and compare the coefficients in the second and third factorial moments and show that
$$
\P(N=2)=\rho^4/2^6 3\sqrt{3}+o(\rho^4).
$$

\subsection{Outline of the proofs}

The proofs of Proposition \ref{expp} and Theorem \ref{sfm} are based on the Kac-Rice formula applied to the gradient of $\Psi$. The Kac-Rice formula is a standard tool for computing the expected number and higher (factorial) moments of the zero set of a Gaussian field.
In general, under some non-degeneracy conditions on the given random field, for every $n\ge 1$ the factorial moments are given by:
\begin{align}
\label{KRn}
\mathbb{E}[ {\cal N}^c_{\rho} ({\cal N}^c_{\rho}-1) \cdots ({\cal N}^c_{\rho}-(n-1))]= \idotsint_{{\cal B}(\rho) \times  \cdots \times {\cal B}(\rho) } K_n({\bf z}) \; d {\bf z},
\end{align}
where ${\bf z}=(z_1, \dots z_n) \in {\cal B}(\rho) \times \cdots \times {\cal B}(\rho) \subset \mathbb{R}^{2n}$, and $K_n$ is the $n$-point correlation function defined as the conditional Gaussian expectation
\begin{align} \label{kkn}
K_n({\bf z})=\phi_{(\nabla \Psi(z_1), \dots, \nabla \Psi(z_n))} (0, \dots, 0) \cdot {\mathbb E} \left[ \prod_{i=1}^n |\text{det} H_{\Psi} (z_i)| \;  \big| \nabla \Psi(z_1)= \cdots= \nabla \Psi(z_n)=0  \right],
\end{align}
where $\phi_{(\nabla \Psi(z_1), \dots, \nabla \Psi(z_n))}(0, \dots, 0)$ is the density function of the Gaussian vector
$$(\nabla \Psi(z_1), \dots, \nabla \Psi(z_n))$$ evaluated at $(0, \dots, 0)$, and $H_{\Psi} (z_i)$ is the Hessian matrix of $\Psi$ at $z_i$. The Kac-Rice formula in \eqref{KRn} holds under the condition that the Gaussian vector  $(\nabla \Psi(z_1), \dots,  \nabla \Psi(z_1))$ is non degenerate.

For $n=1$ the computation of $K_1$ is straightforward; it is essentially the same as in \cite{CaMaWi}. We give it below since
demonstrates the use of the Kac-Rice formula.
The case $n=2$ is more involved and the asymptotics of $K_{2}(z_{1},z_{2})$ as $z_2\to z_1$  (inducing on the second factorial moment)
was entirely unexpected. Based on the above computer simulations one would expect for the critical points repel, i.e. as $z_2\to z_1$, $K_2(z_1,z_2)\to 0$. That would have indicated that the second factorial moment is $o(\rho^4)$, with plausible true order $\rho^5$ of decay. To our surprise, a precise analysis of the relevant Gaussian integrals have shown that $K_2$ does not vanish on the diagonal; it has a finite, non-zero limit. Hence the second factorial moment for small $\rho$ behaves like a constant times the square of the area of ${\cal B}(\rho)$, i.e. a constant times $\rho^{4}$.

It is theoretically possible to compute the behaviour of the higher correlation functions $K_n$ near the diagonal i.e. when $z_i\to z_j$ for $i\ne j$, but seems extremely technically demanding. On the other hand, it is easy to believe, that $K_n$ should stay bounded. Considering it as a given and using the same argument as in the proof of Corollary \ref{prob012}, we obtain the following conjecture.

\begin{conjecture}
\label{uno}
%For every $n$ and for $\rho$ sufficientely small, the $n$-point correlation function $K_n$ is bounded by a constant, depending on $n$ only, in  ${\cal B}(\rho) \times \cdots \times {\cal B}(\rho)$.
For $n>2$ and $\rho\to 0$ we have the following estimate of the factorial moment
$$
\mathbb{E}[ {\cal N}^c_{\rho} ({\cal N}^c_{\rho}-1) \cdots ({\cal N}^c_{\rho}-(n-1))]=O(\rho^{2n})
$$
and, correspondingly, on the probability to have exactly $n$ points in a small ball
$$
\P({\cal N}^c_{\rho}=n)=O(\rho^{2n}).
$$
\end{conjecture}
Be believe that this estimate holds, but we know that it is not sharp since already for $n=3$ we have  $\P({\cal N}^c_{\rho}=3)=O(\rho^{7}\log(1/\rho))=o(\rho^6)$ (see \eqref{eq:three point probability}).

\subsection{Acknowledgements}
The research leading to these results has received funding from the European
Research Council under the European Unions Seventh Framework Programme (FP7/2007-2013) / ERC grant
agreements no 335141 (I.W. and V.C.) and Engineering \& Physical Sciences Research Council (EPSRC) Fellowship EP/M002896/1 (D.B.). We are grateful to Peter Sarnak for some stimulating discussions especially with regards to the proof of Theorem \ref{sfm}, and to Robert Adler, Anne Estrade, and Mikhail Sodin for useful conversations.

\section{Expected number of critical points}\label{exp}

\subsection{On the Kac-Rice formula for computing the expected number of critical points}
The Kac-Rice formula is a standard tool for studying the expected number of zeros of a process (see e.g. \cite[Theorem 11.2.1]{Adler-Taylor} or \cite[Theorem 6.8]{AzWsch}), and its higher moments by expressing the $n$-th (factorial) moment in terms of an $n$-dimensional integral. We apply Kac-Rice formula in this section to compute the expected value of ${\cal N}_{\rho}^c$.

Counting the critical points of $\Psi$ in the ball ${\cal B}(\rho)$ is equivalent to counting the
zeros of the map ${\cal B}( \rho) \to {\mathbb R}^2$ given by $z \to \nabla \Psi(z)$. One defines the zero density
$K_1: {\cal B}(\rho) \to \mathbb{R}$ of $\Psi$ as
$$K_1(z)=\phi_{\nabla \Psi(z)} (0,0) \cdot \mathbb{E}[|\mathrm{det}\, H_{\Psi}(z)| \big| \nabla \Psi(z)=0],$$
where $\phi_{\nabla \Psi(z)}$ is the Gaussian probability density of 2-dimensional vector $\nabla \Psi(x) \in \mathbb{R}^2$
evaluated at $(0,0)$, and $H_{\Psi}(z)$ is the Hessian matrix of $\Psi$ at $z$. By the Kac-Rice formula, if $\nabla \Psi(z)$ is nonsingular for all $z \in {\cal B}(\rho)$, then
\begin{align} \label{KRe}
\mathbb{E}[{\cal N}^c_{\rho}]=\int_{\cal{B}(\rho)} K_1(z) d z.
\end{align}

\subsection{Proof of Proposition \ref{expp}}

We first observe that in our case the zero density $K_1$ is independent of $z$ because $\Psi$ is isotropic; hence the Kac-Rice formula \eqref{KRe} sates that
\begin{align} \label{cont3}
\mathbb{E}[{\cal N}^c_{\rho}]=\pi \rho^2 K_1.
\end{align}
Moreover, as we are dealing with a smooth Gaussian field, it is possible to write an analytic expressions for $K_1$ in terms of the covariance function $\psi$ and its derivatives;  to derive such analytic expression we evaluate the covariance matrix $\Sigma$ of the $5$-dimensional centred jointly Gaussian vector
$$(\nabla \Psi(z), \nabla^2 \Psi(z))$$
where $ \nabla^2 \Psi(z)$ is the vectorized Hessian evaluated at $z$, that is a vector $$(\partial^2_{z_1,z_1}\Psi(z),\partial^2_{z_1,z_2}\Psi(z),\partial^2_{z_2,z_2}\Psi(z)).$$ The covariance matrix $\Sigma$ of $(\nabla \Psi(z), \nabla^2 \Psi(z))$ is evaluated in Appendix \ref{cov-exp} and has the form
$$\Sigma=\left(\begin{array}{ccc}
A  & B\\
 B^t& C
\end{array} \right),$$
where
\begin{align} \label{matrices}
A= \left(\begin{array}{cc}
\frac{1}{2} &0 \\
0&\frac{1}{2}
\end{array} \right), \hspace{1.5cm} B=0, \hspace{1.5cm}  C= \left( \begin{array}{ccc}
\frac{3 }{8} &0& \frac{1}{8} \\
0&\frac{1}{8}&0\\
\frac{1}{8}&0&\frac{3}{8}
\end{array}\right).
\end{align}
From $A$ we immediately obtain the probability density of the $2$-dimensional vector $\nabla \Psi(z)$ evaluated at $(0,0)$:
\begin{align} \label{cont2}
\phi_{\nabla \Psi(x)} (0,0)=\frac{1}{2 \pi \sqrt{ 1/4} }=\frac{1}{\pi},
\end{align}
in addition, since the first and the second order derivatives of $\Psi$ are independent at every fixed point $z \in \mathbb{R}^2$, we have that
$$
\mathbb{E}[|\mathrm{det} H_{\Psi}(x)| \big| \nabla \Psi(x)=0]=\mathbb{E}[|\mathrm{det} H_{\Psi}(x)|].
$$
From the covariance matrix $C$ of $\nabla^2 \Psi(z)$ in \eqref{matrices} we immediately see that
\begin{align} \label{ev}
\mathbb{E}[|\mathrm{det} H_{\Psi}(x)|]= \frac{1}{8} \mathbb{E}[|Y_1 Y_3 -Y_2^2|],
\end{align}
where $Y=(Y_{1},Y_{2},Y_{3})$ is a centred jointly Gaussian random vector with covariance matrix
\begin{equation*}
C_1=\left(
\begin{array}{ccc}
3 & 0 & 1 \\
0 & 1 & 0 \\
1 & 0 & 3
\end{array}
\right).
\end{equation*}
To evaluate \eqref{ev} we introduce the transformation $W_1=Y_1$, $W_2=Y_2$, $W_3=Y_1+Y_3$, and we write $\mathbb{E}[|Y_1 Y_3 -Y_2^2|]$ in terms of a conditional expectation as follows
\begin{align} \label{ccc}
 \mathbb{E}[|Y_1 Y_3 -Y_2^2|]= \mathbb{E}_{W_3} [   \;  \mathbb{E}[|W_1 W_3-W_1^2 - W_2^2 | \big| W_3=t  ] \;  ];
\end{align}
to evaluate the conditional expectation in \eqref{ccc} we follow the argument in the proof of
\cite[Proposition 1.1]{CaMaWi}, i.e. we note that
\begin{align*}
 \mathbb{E}[|W_1 W_3-W_1^2 - W_2^2 |  \big| W_3=t]&
=\mathbb{E}\left[ \left. \left\vert   W_{1} \, t -W_{1}^{2}-W_{2}^{2}%
\right\vert \right\vert W_{3}= t \right]  \\
&=\mathbb{E}\left[ \left\vert   (Z_{1}+ {t}/{2})\, t -(Z_{1}+ {t}/{2})^{2}-Z_{2}^{2}\right\vert \right]  \\
&=\mathbb{E}\left[ \left\vert -Z_{1}^{2}-Z_{2}^{2}+2t^{2}\right\vert \right]=\mathbb{E} \left[   \left|  - X+\frac{t^2}{4}  \right|  \right],
\end{align*}
where $Z_1, Z_2$ are independent standard Gaussian and $X$ is a $\chi$-squared random variable with density
\begin{align*}
f_X(x)=\frac{1}{2}e^{-\frac{x}{2}}, \hspace{2cm} x >0.
\end{align*}
It follows that
\begin{align*}
\mathbb{E} \left[   \left| \frac{t^2}{4} - X  \right|  \right] = -2+4 e^{-\frac{t^2}{8}}+\frac{t^2}{4},
\end{align*}
and
\begin{align} \label{cont1}
\mathbb{E}[|Y_1 Y_3 -Y_2^2|]=\frac{1}{4 \sqrt \pi} \int_{\mathbb{R}} e^{-\frac{t^2}{16}} \left( -2+4 e^{-\frac{t^2}{8}}+\frac{t^2}{4} \right) d t = \frac{2^2}{\sqrt 3 }.
\end{align}
The statement follows combining \eqref{cont3}, \eqref{cont2},  \eqref{cont1}, and observing that
\begin{align*}
\mathbb{E}[{\cal N}^c_{\rho}(\Psi)]=\pi \rho^2 \cdot \frac{1}{\pi} \cdot   \frac{1}{8} \frac{2^2}{\sqrt 3 }= \frac{1}{2 \sqrt 3 }\cdot \rho^2.
\end{align*}

\section{Second factorial moment} \label{fact}

\subsection{On the Kac-Rice formula for computing the second factorial moment of the number of critical points}

We will find an explicit expression for the $2$-point correlation function $K_2 : {\cal B}(\rho) \times {\cal B}(\rho) \to \mathbb{R}$, defined as the conditional Gaussian expectation
\begin{align*}
K_2(z,w)=\phi_{(\nabla \Psi(z),\nabla \Psi(w))} (0, 0) \cdot \mathbb{E}[ |\mathrm{det} H_{\Psi}(z)| \cdot |\mathrm{det} H_{\Psi}(w)| \big| \nabla \Psi(z)=\nabla \Psi(w)=0],
\end{align*}
 in terms of the covariance function $\psi$ and its derivatives. Finding such an expression involves studying the centred Gaussian vector
\begin{equation} \label{repp}
(\nabla \Psi(z),\nabla \Psi(w),\nabla^2 \Psi(z),\nabla^2 \Psi(w))
\end{equation}
with covariance matrix ${\bf \Sigma}(z,w)$, $z,w \in  {\cal B}(\rho)$. It is known \cite[Theorem 6.9]{AzWsch} that, if for all $z \ne w$ the Gaussian distribution of $(\nabla \Psi(z),\nabla \Psi(w))$ is non-degenerate,  the second factorial moment
of the number of critical points in ${\cal B}(\rho)$ can be expressed as
\begin{align} \label{18:03}
\mathbb{E}[{\cal N}^c_{\rho} \;  ({\cal N}^c_{\rho}-1)]=\iint_{\cal{B}(\rho) \times \cal{B}(\rho)} K_2(z,w) \; d z \, d w.
\end{align}
We note that $K_2$ is everywhere nonnegative.

\subsection{Proof of Theorem \ref{sfm}}
\label{subsec: proof of theorem}
In order to study the asymptotic behaviour of the second factorial moment of the number of critical points in ${\cal B}(\rho)$, as the radius $\rho$ of the disk goes to zero,  we need to study the centred Gaussian random vector \eqref{repp}.
%$$(\nabla \Psi(z),\nabla \Psi(w),\nabla^2 \Psi(z),\nabla^2 \Psi(w));$$
Its covariance matrix ${\bf \Sigma}={\bf \Sigma}(z,w)$ is of the form
$${\bf \Sigma}=\left(\begin{array}{ccc}
{\bf A}  & {\bf B}\\
 {\bf B}^t& {\bf C}
\end{array} \right),$$
where
${\bf A}={\bf A}(z,w)$ is the covariance matrix of the gradients $(\nabla \Psi(z),\nabla \Psi(w))$, ${\bf C}={\bf C}(z,w)$ is the covariance matrix of the second order derivatives $(\nabla^2 \Psi(z),\nabla^2 \Psi(w))$ and ${\bf B}={\bf B}(z,w)$ is the covariance matrix of the first and second order derivatives.

The function $\Psi$ is isotropic, hence, the critical point process is also invariant w.r.t. translations and rotations. This means that its $2$-point function $K_2(z,w)$ depends on $|z-w|$ only (this is not true for covariance matrix $\bf \Sigma$); by the standard abuse of notation we write
\begin{equation}
\label{eq:K2=K2(d)}
K_2(z,w)=K_2(|z-w|).
\end{equation}
We will compute $K_2(z,w)$ for $z=(0,0)$ and $w=(0,r)$, which, thanks to the by-product \eqref{eq:K2=K2(d)} of the isotropic property of $\Psi$, this will give us $K_2(r)$.

From now on we will work only with ${\bf \Sigma}(r)$ and ${\bf \Delta}(r)$ which we {\em define} as ${\bf \Sigma}(z,w)$ and ${\bf \Delta}(z,w)$ with $z$ and $w$ as above. In Appendix \ref{ucn} we evaluate the entries of ${\bf \Sigma}(z,w)$ under this assumption, and in Appendix \ref{matrixdelta} we evaluate the covariance matrix ${\bf \Delta}={\bf \Delta}(z,w)$ of $(\nabla^2 \Psi(z),\nabla^2 \Psi(w))$ conditioned on $\nabla \Psi(z)=\nabla \Psi(w)=0$, i.e.,
\begin{align*}
{\bf \Delta}={\bf C} - {\bf B}^t {\bf A}^{-1} {\bf B}.
\end{align*}

As we discussed above, the two-point function is given by
\begin{equation}
\label{k2}
\begin{aligned}
K_2(r)&=\frac{1}{(2 \pi)^2 \sqrt{\text{det}({\bf A}(r))}} \\
&\times \int_{\mathbb{R}^6 }|\zeta_{1} \zeta_{3} - \zeta^2_{2}| \cdot |\zeta_{4} \zeta_{6} - \zeta^2_{5}| \frac{1}{(2 \pi)^3} \frac{1}{\sqrt{\text{det}({\bf \Delta}(r))}} \exp \left\{-\frac{1}{2} \zeta^t {\bf \Delta}^{-1}(r) \zeta \right\} d \zeta,
\end{aligned}
\end{equation}
where $\zeta=(\zeta_1, \zeta_2, \zeta_3, \zeta_4, \zeta_5, \zeta_6)$ is a vector in $\mathbb{R}^6$.
Indeed, the density of $(\nabla \Psi(0,0),\nabla\Psi(0,r))$ at zero is given by $(2\pi)^{-2}(\det({\bf A}(r)))^{-1/2}$, and the integral gives the expectation of $|\mathrm{det} H_{\Psi}(z)| \cdot |\mathrm{det} H_{\Psi}(w)|$ with respect to the Gaussian measure of $(\nabla^2 \Psi(z),\nabla^2\Psi^2(w))$ conditioned on  $\nabla \Psi(z)=\nabla \Psi(w)=0$, that is, having covariance ${\bf \Delta}(r)$.

Our aim is to study the asymptotic behaviour of the $2$-point correlation function $K_2$ in the {\em vicinity} of $r=0$.

For every strictly positive $r$, ${\bf \Delta}(r)$ is symmetric, hence we may diagonalise it with an orthogonal $P(r)$:
\begin{align} \label{18:07}
{\bf \Delta}(r)=P^{-1}(r) \Lambda(r) P(r)=P^t(r) \Lambda(r) P(r),
\end{align}
where the matrix $\Lambda(r)$ is diagonal, with eigenvalues $\lambda_i(r)$, $i=1, \dots,6$, and $P(r)$ is the orthogonal matrix with row vectors the normalized eigenvectors of ${\bf \Delta}(r)$. The analytic expressions of the eigenvalues and eigenvectors of ${\bf \Delta}(r)$, $r>0$, are computed in  Lemma \ref{eigenv} and Lemma \ref{eigenvect} respectively.
In Lemma \ref{eigenvT} and Lemma \ref{lem: expansion for Q} we compute their  Taylor expansion around $r=0$. We prove these lemmas with the aid of Mathematica since the calculations are technically demanding. We stress that all the computations performed with Mathematica are symbolic.
\\

Equation \eqref{18:07} implies that we can write
\begin{equation}
\label{13:05}
\begin{aligned}
&\frac{1}{\sqrt{\text{det}({\bf \Delta}(r))}}\exp \Big\{ -\frac{1}{2} \zeta^t {\bf \Delta}^{-1}(r) \zeta \Big\}&
\\
&=
\frac{1}{\sqrt{\prod_{i=1}^6 \lambda_i(r)}}\exp \Big\{ -\frac{1}{2} \zeta^t P^{-1}(r) \Lambda^{-1}(r) P(r) \zeta \Big\}
\\
&=\frac{1}{\sqrt{\prod_{i=1}^6 \lambda_i(r)}}\exp\Big\{ -\frac{1}{2} (\Lambda^{-1/2}(r)P(r) \zeta)^t  (\Lambda^{-1/2}(r)P(r) \zeta) \Big\} .
\end{aligned}
\end{equation}

This suggests to introduce a new variable $\xi=\Lambda^{-1/2}(r) P(r) \zeta$. Clearly, we can express $\zeta$ in terms of $\xi$ as
\begin{equation}
\label{eq: zeta}
\zeta=P^{-1}(r)\Lambda^{1/2}(r)\xi=P^{t}(r)\Lambda^{1/2}(r)\xi
\end{equation}
With this change of variables
$$
\frac{1}{\sqrt{\text{det}({\bf \Delta}(r))}}\exp \Big\{ -\frac{1}{2} \zeta^t {\bf \Delta}^{-1}(r) \zeta \Big\}d \zeta= e^{-|\xi|^2/2}d\xi.
$$
Using \eqref{eq: zeta}, we can write components $\zeta_i$  as
$$
\zeta_i=\sum_{j=1}^6 (Q(r))_{ij}\; \sqrt{\lambda_j(r)} \; \xi_j=\sum_{j=1}^6 q_{ij}(r)\; \sqrt{\lambda_j(r)} \; \xi_j,
$$
where the $q_{i j}(r)$ are the elements of $Q(r)=P^{-1}(r)=P^{t}(r)$. The columns of $Q$ form an orthonormal basis of eigenvectors of eigenvectors of ${\bf \Delta}(r)$.
With this change of variables  we can rewrite the two quadratic forms $\zeta_{1} \zeta_{3}- \zeta_{2}^2$ and $\zeta_{4} \zeta_{6}- \zeta_{5}^2$ in  \eqref{k2} as
\begin{align*}
\zeta_{1} \zeta_{3} - \zeta^2_{2} &
=\br{\sum_{j=1}^6 q_{1 j}(r) \sqrt{\lambda_{j}(r) } \; \xi_j}
 \br{\sum_{j=1}^6 q_{3 j}(r) \sqrt{\lambda_{j}(r) } \; \xi_j} - \br{\sum_{j=1}^6 q_{2 j}(r) \sqrt{\lambda_{j}(r) } \; \xi_j }^2,
\\
\zeta_{4} \zeta_{6} - \zeta_{5}^2 &
= \br{\sum_{j=1}^6 q_{4 j}(r) \sqrt{\lambda_{j}(r) } \; \xi_j}
\br{\sum_{j=1}^6 q_{6 j}(r) \sqrt{\lambda_{j}(r) } \; \xi_j} - \br{\sum_{j=1}^6 q_{5 j}(r) \sqrt{\lambda_{j}(r) } \; \xi_j }^2.
\end{align*}
Summing it all up, the $2$-point correlation function $K_2$ in \eqref{k2} in $\xi$ coordinates becomes
\begin{equation}
\label{k2ch}
\begin{aligned}
&K_2(r)= \frac{1}{ (2\pi)^5 \sqrt{ \text{det}({\bf A}(r)) }} \\
&\;\; \times
\bigints\limits_{\R^6}
\left|
\br{\sum_{j=1}^6 q_{1 j}(r) \sqrt{\lambda_{j}(r) } \; \xi_j}
\br{\sum_{j=1}^6 q_{3 j}(r) \sqrt{\lambda_{j}(r) } \; \xi_j} - \br{\sum_{j=1}^6 q_{2 j}(r) \sqrt{\lambda_{j}(r) } \; \xi_j }^2
\right|
\\
&\;\; \times
\left|
\br{\sum_{j=1}^6  q_{4 j}(r) \sqrt{\lambda_{j}(r) } \; \xi_j}
\br{\sum_{j=1}^6  q_{6 j}(r) \sqrt{\lambda_{j}(r) } \; \xi_j} -
\br{\sum_{j=1}^6  q_{5 j}(r) \sqrt{\lambda_{j}(r) } \; \xi_j}^2
\right|
\\
& \;\; \times \exp \Big\{-\frac{1}{2} \sum_{i=1}^6 \xi^2_i \Big\} d \xi.
\end{aligned}
\end{equation}

To obtain the asymptotic behaviour {\em around} $r=0$ of the integral in \eqref{k2ch}, we Taylor expand around the origin the entries $q_{i j}$ of the  matrix $Q$ and eigenvalues $\lambda_j$.  Such Taylor expansions up to $O(r^4)$ are given by equations \eqref{eq: expansion for Q} and \eqref{eq:lambda series}. Combining these expansions and noting that the first two factors in the integrand are polynomials of degree $4$ in terms of $\xi$ we obtain the following expansion:
\begin{align*}
&\Big[ \sum_{j} q_{1 j}(r) \sqrt{\lambda_{j}(r) } \; \xi_j \sum_{j} q_{3 j}(r) \sqrt{\lambda_{j}(r) } \; \xi_j - \Big(\sum_{j} q_{2 j}(r) \sqrt{\lambda_{j}(r) } \; \xi_j \Big)^2 \Big] \\
&\;\; \times \Big[ \sum_{j} q_{4 j}(r) \sqrt{\lambda_{j}(r) } \; \xi_j \sum_{j} q_{6 j}(r) \sqrt{\lambda_{j}(r) } \; \xi_j - \Big(\sum_{j} q_{5 j}(r) \sqrt{\lambda_{j}(r) } \; \xi_j \Big)^2 \Big]\\
&= - \frac{1}{2^7 3} \xi_4^2 \xi_6^2 r^2 +   (1+ ||\xi||^4)\; O(r^4 ),
\end{align*}
and then
\begin{align} \label{14:00}
K_2(r)%&= \frac{1}{ (2\pi)^5 \sqrt{ \text{det} A(r) }} \iint_{\mathbb{R}^3 \times \mathbb{R}^3}
%\big| h(r,\xi) \big| \times \exp \left\{-\frac{1}{2} \sum_{i=1}^6 \xi^2_i \right\} d \xi_1 \cdots d \xi_6\\
&= \frac{1}{ (2\pi)^5 \sqrt{ \text{det}({\bf A}(r)) }} \left[ \frac{r^2}{2^7 3} \int_{\mathbb{R}^6}
\xi_4^2 \xi_6^2 \times \exp \left\{-\frac{1}{2} \sum_{i=1}^6 \xi^2_i \right\} d \xi+ O(r^4) \right].
\end{align}\\
In the Gaussian integral variables separate and it is a product of standard one-dimensional integrals. Each of them is equal to $\sqrt{2\pi}$ and the entire integral is $(2\pi)^3$.
Matrix $\bf A$ has a simple block structure and it is easy to compute its determinant. Explicit computation in Appendix \ref{ucn} (see equation \eqref{14:25}) gives
$$
\det(A)=\frac{3r^4}{2^8}+O(r^6).
$$
Combining this asymptotic with \eqref{14:00}, we finally obtain that, as $r \to 0$,
\begin{align*}
K_2(r)=\frac{1}{2^5 3 \sqrt 3 \pi^2}+O(r^2),
\end{align*}
and, in view of \eqref{18:03}, as $\rho \to 0$,
\begin{align*}
\mathbb{E}[{\cal N}^c_{\rho} \; ({\cal N}^c_{\rho}-1)]%= \pi \rho^2 2 \pi \int_{0}^{\rho}K_2(r) r d r
= \frac{1}{2^5 3 \sqrt 3 \pi^2} \pi^2 \rho^4+O(\rho^6)=\frac{1}{2^5 3 \sqrt 3 }  \rho^4+O(\rho^6).
\end{align*}

To prove the second part of Theorem \ref{sfm} we need to evaluate the two-point correlation function $K_2$ modified for the respective types of critical points. The modified function $K_{2}$ has the same expression \eqref{k2} with the integration over a proper subset of $\R^6$, i.e.
the $\zeta$ with the corresponding critical points of the prescribed types.

To be more precise, let us define two Hessians at points $z$ and $w$ (already conditioned to be are critical points). In terms of $\zeta_i$ these Hessians are given by
\begin{equation*}
H_1=
\begin{pmatrix}
\zeta_1 & \zeta_2 \\
\zeta_2 & \zeta_3
\end{pmatrix},
\quad  \mathrm{and} \quad
H_2=
\begin{pmatrix}
\zeta_4 & \zeta_5 \\
\zeta_5 & \zeta_6
\end{pmatrix}
\end{equation*}
The characteristic polynomials for these matrices are $$
x^2+b_1x+c_1=x^2-(\zeta_1+\zeta_3)x+\zeta_1\zeta_3-\zeta_2^2
$$
and
$$
x^2+b_2x+c_2=x^2-(\zeta_4+\zeta_6)x+\zeta_4\zeta_6-\zeta_5^2.
$$
The particular type of a critical point depends on the eigenvalues of its Hessian: they are both negative for the local maxima, negative for the local minima, and of different signs for the saddles. We may reformulate these dependencies in terms of the coefficients $b_i=-\mathrm{ Tr}\, H_i$ and $c_i=\det(H_i)$: a critical point with Hessian $H_i$ is a minimum if $c_i>0$ and $b_i<0$, a maximum if $c_i>0$ and $b_i>0$, and a saddle if $c_i<0$ (we may ignore the special probability $0$ cases when one of the eigenvalues vanishes).

As before, we rewrite
$\zeta_i$ in terms of of $\xi_i$. This gives the coefficients of the polynomials as functions of $\xi_i$ and $r$. Expanding in powers of $r$ we get
\begin{equation}
\label{eq:b1,c1,b2,c2 as on r}
\begin{aligned}
b_1&=-\frac{\xi_6}{\sqrt{3}}+\br{\frac{\xi_6}{144\sqrt{3}}
-\frac{\xi_5}{96\sqrt{2}}}r^2+O(r^3)=b_{1,0}+b_{1,2}r^2+O(r^3)
\\
c_1&=-\frac{\xi_4\xi_6}{8\sqrt{6}}r+
\frac{-9\xi_1^2-9\xi_4^2+2\sqrt{6}\xi_6\xi_5+4\xi_6^2}{2^7 3^2}r^2+O(r^3)
=c_{1,1}r+c_{1,2}r^2+O(r^3)
\\
b_2&=-\frac{\xi_6}{\sqrt{3}}+\br{\frac{\xi_6}{144\sqrt{3}}
-\frac{\xi_5}{96\sqrt{2}}}r^2+O(r^3)=b_{2,0}+b_{2,2}r^2+O(r^3)\\
c_2&=\frac{\xi_4\xi_6}{8\sqrt{6}}r+
\frac{-9\xi_1^2-9\xi_4^2+2\sqrt{6}\xi_6\xi_5+4\xi_6^2}{2^7 3^2}r^2+O(r^3)
=c_{2,1}r+c_{2,2}r^2+O(r^3).
\end{aligned}
\end{equation}
We observe the following: all of the coefficients $b_{i,j}$ are linear functions of $\xi$, and all of the coefficients 
$c_{i,j}$ are quadratic forms. We also notice that
$$
b_{1,0}=b_{2,0}, \ \
b_{1,2}=b_{2,2}, \  \
c_{1,1}=-c_{2,1}, \  \
c_{1,2}=c_{2,2}.
$$

Since all the expressions we deal with are homogeneous functions of various degrees, it is natural to work in spherical coordinates. We introduce $s_i=\xi_1/|\xi|$ and rescale $b_i$ by $|\xi|$ and $c_i$ by $|\xi|^2$. Abusing notation we denote the rescaled coefficients $b_i$, $b_{i,j}$, $c_i$, and $c_{i,j}$ that are now functions of $s_i$ instead of $\xi_i$ by the same letters; there is no confusion since from now on all expressions will be in terms of $|\xi|\in (0,\infty)$ and $s=(s_1,\dots,s_6)\in S^5$.
With this notation, the formula \eqref{k2ch} for $K_2$ becomes
\begin{equation}
\label{eq:spherical coordinates}
\begin{aligned}
K_{2}(r)&= \frac{1}{ (2\pi)^5 \sqrt{ \text{det}({\bf A}(r)) }}
 \int_{\R^6} |\xi|^4 |c_1 c_2| e^{-|\xi|^2/2}d\xi
 \\
& =
  \frac{1}{ (2\pi)^5 \sqrt{ \text{det}({\bf A}(r)) }}
\int_0^\infty |\xi|^9 e^{-|\xi|^2/2} d |\xi|
\int_{S^5}|c_1(s) c_2(s)| ds
\\&= \frac{12}{ \pi^5 \sqrt{ \text{det}({\bf A}(r)) }}
\int_{S^5}|c_1(s) c_2(s)| ds,
\end{aligned}
\end{equation}
where $ds$ is the spherical volume element on the unit sphere $S^5$, and we evaluated the standard Gaussian integral
$$\int_0^\infty |\xi|^9 e^{-|\xi|^2/2} d |\xi| = 2^{7}\cdot 3.$$

\bigskip

\noindent{\bf Minimum-minimum two point function. }

The the two-point correlation function $K_2^{min,min}(r)$ corresponding to the local minima
is given by \eqref{eq:spherical coordinates} except that we replace the domain of the integration $S^5$, by
$$
S_{min,min}=\{s\in S^5: c_1>0, c_2 >0, b_1<0, b_2<0\},
$$
the set of $s$ such that both Hessians correspond to local minima.
If $|s_4 s_6|>\c r$ for sufficiently large $\c$, then $c_1$ and $c_2$ are of opposite signs (for the rest of this section we use  $\c$ to denote all absolute constants). This implies that $S_{min,min}$ is a subset of $\{s:|s_4 s_6|<\c r\}$ for some $\c$ sufficiently big. It is easy to see that on this set $|c_i|=O(r^2)$, thus
$$
\int_{S_{min,min}}|c_1(s) c_2(s)| ds
\le
\int_{\{s:|s_4 s_6|<\c r\}}|c_1(s) c_2(s)| ds
\le
O(r^4)\int_{\{s:|s_4 s_6|<\c r\}}ds
=
O(r^5\log(1/r)).
$$

That yields $K_2^{min,min}(r)
=O(r^3\log(1/r))$ via \eqref{eq:spherical coordinates}, where $r^2$ cancelled out with $\sqrt{\det({\bf A})}$). Integrating this estimate over $\cal{B}(\rho)\times \cal{B}(\rho)$ we obtain an estimate of the second factorial moment:
$$
\E\brb{{\cal N}_\rho^{min}({\cal N}_\rho^{min}-1)}=O(\rho^7\log(1/\rho)).
$$
The other estimate of \eqref{eq:max-max moment} follows from symmetry considerations.

\bigskip

\noindent{\bf Minimum-maximum two point function. }

In the similar way, we have to estimate the integral over $S_{min,max}$, the set where one point is a minimum and the other is a maximum.  This set is given by conditions that both $c_i$ are positive and $b_1$ and $b_2$ are of different signs.
First, the same argument as above forces $|s_4 s_6|<\c r$ for some large constant $\c$. If $|s_6|>\c r^2$ for a large constant $\c$, then the leading terms in the formulas \eqref{eq:b1,c1,b2,c2 as on r} for $b_i$ dominate and $b_1$ and $b_2$ are of the same sign, contradicting our assumption. Hence this implies that $|s_6|<\c r^2$, and under this assumption both $b_i$ are of the form
$$
-\frac{s_6}{\sqrt{3}}-\frac{s_5 r^2}{96\sqrt{2}}+O(r^3).
$$
Again, since $b_i$ should be of different signs, it forces the term corresponding to $O(r^3)$ to dominate, that is
$$
L(s_5,s_5)=\left|-\frac{s_6}{\sqrt{3}}-\frac{s_5 r^2}{96\sqrt{2}}\right|\le \c r^3,
$$
for some big constant $\c$. Notice that this condition is stronger than the previous condition that $|s_6|<\c r^2$.

Substituting the estimate $|s_6|<\c r^2$ into the formulas \eqref{eq:b1,c1,b2,c2 as on r} for $c_1$ and $c_2$ we see that they both are equal to
$$
-\frac{9}{2^7 3^2}(s_1^2+s_4^2)r^2+O(r^3).
$$
It then follows that $(s_1^2+s_4^2)$ is bounded by $\c r$ (for large $\c$), as otherwise both $c_1$ and $c_2$ are negative; also under this condition both $c_i$ are $O(r^3)$.
Combining all of this we get the estimate
$$
\int\limits_{S_{min,max}}|c_1(s) c_2(s)| ds
\le
\int\limits_{\substack{(s_1^2+s_4^2)<\c r \\ L(s_5,s_5)<\c r^3 }}|c_1(s) c_2(s)| ds
=
O(r^6)\int\limits_{\substack{(s_1^2+s_4^2)<\c r \\ L(s_5,s_5)<\c r^3 }}ds=O(r^6)O(r^4)=O(r^{10}),
$$
and substituting this into \eqref{eq:spherical coordinates} and integrating $K_{2}$ the Kac-Rice formula yields
$$
\E\brb{\cal{N}_\rho^{min}\cal{N}_\rho^{max}}=O(\rho^{12}).
$$

\bigskip

\noindent{\bf Saddle-saddle and extremum-extremum two point functions. }

For two extrema or two saddle points both $c_i$ are forced to be of the same sign. 
The same argument as for the minimum-minimum case yields
$$
\E\brb{\cal{N}_\rho^{saddle}(\cal{N}_\rho^{saddle}-1)}= O(\rho^7\log(1/\rho)),
$$
and
$$
\E\brb{\cal{N}_\rho^{e}(\cal{N}_\rho^{e}-1)}= O(\rho^7\log(1/\rho)).
$$

\bigskip
\noindent{\bf Extremum-saddle two point function. }

Finally we notice that $\cal{N}=\cal{N}^{e}+\cal{N}^{saddle}$, and
$$
\cal{N}(\cal{N}-1)=\cal{N}^{e}(\cal{N}^{e}-1)+\cal{N}^{saddle}(\cal{N}^{saddle}-1)+2\cal{N}^{e}\cal{N}^{saddle}.
$$
Combining this formula with previous estimates we obtain
$$
\E\brb{\cal{N}_\rho^{e}\cal{N}_\rho^{saddle}}=\frac{1}{2}\E\brb{\cal{N}_\rho(\cal{N}_\rho-1)}+O(\rho^7\log(1/\rho))
=
\frac{1}{2^6 \, 3 \, \sqrt 3 } \; \rho^4+O(\rho^6).
$$
This completes the proof of Theorem \ref{sfm}.

\appendix

\section{\texorpdfstring{Eigenvalue and eigenvectors of ${\bf \Delta}(r)$, $r>0$}
{Eigenvalue and eigenvectors of Delta, r>0}
}

We introduce the notation
$$
{\bf \Delta} (r)=\left( \begin{array}{cc}  \Delta_1(r) &  \Delta_2(r) \\ \Delta_2(r) & \Delta_1(r)
\end{array} \right)
$$
where $\Delta_1$ and $\Delta_2$ are $3\times 3$ symmetric matrices with elements in terms of $a_i$, $i=1,\dots 8$
\begin{equation}
\label{eq:definition delta_i}
\begin{aligned}
 \Delta_1(r)
= \left( \begin{array}{ccc}
\frac{1}{3} +a_1(r)&0& a_4(r)\\
0&a_2(r)&0\\
a_4(r)&0&a_3(r)
\end{array}\right), \hspace{1cm}
\Delta_2(r)
= \left( \begin{array}{ccc}
\frac 1 3 + a_5(r)&0& a_8(r)\\
0&a_6(r)&0\\
a_8(r)&0&a_7(r)
\end{array}\right).
\end{aligned}
\end{equation}

\label{eigen}

\noindent We compute now the eigenvalues and eigenvectors of the matrix ${\bf \Delta}(r)$, $r>0$. We introduce the following notation:

\begin{align*}
A_1^{+}(r)&= a_1(r) + a_5(r)+\frac{2}{3}, \hspace{1cm} A_1^{-}(r)= a_1(r) - a_5(r),
\end{align*}
\begin{align*}
A_2^{\pm}(r)&= a_2(r) \pm a_6(r), \\
A_3^{\pm}(r)&= a_3(r) \pm a_7(r), \\
A_4^{\pm}(r)&= a_4(r) \pm a_8(r);
\end{align*}
with $a_i(r)$  defined above.

\begin{lemma} \label{eigenv}
For every $r >0$, the eigenvalues of the matrix ${\bf \Delta}(r)$ have the following explicit expressions:
\begin{equation}
\label{eq:lambda explicit}
\begin{aligned}
\lambda_1(r)&=A_2^{-}(r),\\
\lambda_2(r)&=A_2^{+}(r), \\
\lambda_3(r)&=\frac{1}{2} \Big[A_1^{-}(r) +A_3^{-}(r) - \sqrt{(A_1^{-}(r) - A_3^{-}(r))^2+4 A_4^{-}(r)^2} \Big], \\
\lambda_4(r)&=\frac{1}{2} \Big[A_1^{-}(r) +A_3^{-}(r) + \sqrt{(A_1^{-}(r) - A_3^{-}(r))^2+4 A_4^{-}(r)^2} \Big],\\
\lambda_5(r) &=\frac{1}{2} \Big[A_1^{+}(r) +A_3^{+}(r) - \sqrt{(A_1^{+}(r) - A_3^{+}(r))^2+4 A_4^{+}(r)^2} \Big], \\
\lambda_6(r)&=\frac{1}{2} \Big[A_1^{+}(r) +A_3^{+}(r) + \sqrt{(A_1^{+}(r) - A_3^{+}(r))^2+4 A_4^{+}(r)^2} \Big].
\end{aligned}
\end{equation}
\end{lemma}

\begin{proof}
We can compute explicitly the roots of
$$
\text{det}({\bf \Delta}(r)-\lambda I)=\text{det} \left( \begin{array}{cc} \Delta_1(r) - \lambda I & \Delta_2(r) \\ \Delta_2(r) & \Delta_1(r) -\lambda I
\end{array} \right),
$$
by observing that since $\Delta_i$ are square matrices, we have the following identity for the determinant of a block matrix
$$
\text{det} \left( \begin{array}{cc} \Delta_1(r) - \lambda I & \Delta_2(r) \\ \Delta_2(r) & \Delta_1(r) -\lambda I
\end{array} \right) = \text{det} (\Delta_1(r) - \lambda I - \Delta_2(r)) \, \text{det} (\Delta_1(r) - \lambda I + \Delta_2(r)).
$$
The matrices $\Delta_1(r) - \lambda I \pm \Delta_2(r)$ could be written in terms of $A_i^\pm$ as
\begin{align*}
\Delta_1(r) - \lambda I \pm \Delta_2(r)= \left(\begin{array}{ccc}
A_1^{\pm}(r)-\lambda & 0 & A_4^{\pm}(r) \\
0 & A_2^{\pm}(r) - \lambda &0 \\
A_4^{\pm}(r) & 0 & A_3^{\pm}(r)-\lambda
\end{array} \right).
\end{align*}
Since these matrices have many elements equal to zero, their determinants are particularly simple and could be factorized as
\begin{align*}
\text{det} (\Delta_1(r) - \lambda I \pm \Delta_2(r))= (A_2^{\pm}(r)-\lambda) [\lambda^2 - \lambda (A_1^{\pm}(r)+A_3^{\pm}(r))
+ A_1^{\pm}(r) A_3^{\pm}(r) - A_4^{\pm}(r)^2].
\end{align*}
The last factor is quadratic in terms of $\lambda$ and the roots could be found explicitly. They are equal to $\lambda_5$ and $\lambda_6$ in the ``$+$'' case and $\lambda_3$ and $\lambda_4$ in the ``$-$'' case.
\end{proof}

This lemma expresses the eigenvalues of $\bf \Delta$ in terms of $A_i^\pm$ which, in their term, are expressed in terms of $a_i$. In \eqref{eq: series for a} we will compute the asymptotic behaviour of $a_i$. Substituting these expansions into explicit formulas (using Mathematica)\eqref{eq:lambda explicit} we get expansions for $\lambda_i$ and $\sqrt{\lambda_i}$.

\begin{lemma} \label{eigenvT}
The following Taylor expansions hold around the origin
\begin{equation}
\label{eq:lambda series}
\begin{aligned}
\lambda_1(r)&=\frac{r^2}{2^6}- \frac{7}{2^{14} 3^2 5} r^6+O(r^8), &
\sqrt{\lambda_1(r)}&=\frac{r}{4\sqrt{2}}+O(r^5)\\
\lambda_2(r)&= \frac{r^4}{2^{10} 3^2}-\frac{r^6}{2^{13} 3^3 5}+O(r^8), &
\sqrt{\lambda_2(r)}&=\frac{r^2}{2^5 3}+O(r^4) \\
\lambda_3(r)&=O(r^8), &
\sqrt{\lambda_3(r)}&=O(r^4)\\
\lambda_4(r)&=\frac{r^2}{2^5}+ \frac{37}{2^{13} 3^2 5}r^6+O(r^8), &
\sqrt{\lambda_4(r)}&=\frac{r}{4\sqrt{2}}+O(r^5)\\
\lambda_5(r) &= \frac{r^4}{2^{10} 3^2} + \frac{7}{2^{14} 3^3 5}r^6 +O(r^8),  &
\sqrt{\lambda_5(r)}&=\frac{r^2}{2^5 3}+O(r^4)\\
\lambda_6(r)&=\frac{2}{3}- \frac{5}{2^3 3^3}r^2+ \frac{191}{2^{10} 3^5}r^4- \frac{11 \cdot 241}{2^{14} 3^7 5}r^6+O(r^8)  &
\sqrt{\lambda_6(r)}&=\frac{\sqrt{2}}{\sqrt{3}}- \frac{5r^2}{2^4 3^2\sqrt{6}}+O(r^4).
\end{aligned}
\end{equation}

\end{lemma}

After obtaining the explicit formulas for eigenvalues we, again, use computer algebra to find explicit formulas for eigenvectors  of $\bf \Delta$.

\begin{lemma} \label{eigenvect}
For every $r>0$, the following vectors $v_i(r)$ are the eigenvectors of the matrix ${\bf \Delta}(r)$ corresponding to $\lambda_i(r)$
\begin{equation}
\label{eq:explicit eigenvectors}
\begin{aligned}
v_1(r)&=(0,-1,0,0,1,0),\\
v_2(r)&=(0,1,0,0,1,0), \\
v_3(r)&=( v_{3 1}(r) , 0,-1,- v_{3 1}(r), 0,1), \\
v_4(r)&=( v_{4 1}(r) ,0,-1,- v_{4 1}(r) ,0,1), \\
v_5(r)&=(- v_{5 1}(r) ,0,1, - v_{5 1}(r) ,0,1), \\
v_6(r)&=(- v_{6 1}(r) ,0,1,- v_{6 1}(r) ,0,1).
\end{aligned}
\end{equation}
where
\begin{equation*}
\begin{aligned}
v_{3 1}(r)&= \frac{A_3^-(r) - A_1^-(r) + \sqrt{ [A_3^-(r) - A_1^-(r) ]^2 + 4 A_4^-(r)^2 }}{2 A_4^-(r)}, \\
v_{4 1}(r)&= \frac{A_3^-(r) - A_1^-(r) - \sqrt{ [A_3^-(r) - A_1^-(r) ]^2 + 4 A_4^-(r)^2 }}{2 A_4^-(r)}, \\
v_{5 1}(r)&= \frac{A_3^+(r) - A_1^+(r) + \sqrt{ [A_3^+(r) - A_1^+(r) ]^2 + 4 A_4^+(r)^2 }}{2 A_4^+(r)}, \\
v_{6 1}(r)&= \frac{A_3^+(r) - A_1^+(r) - \sqrt{ [A_3^+(r) - A_1^+(r) ]^2 + 4 A_4^+(r)^2 }}{2 A_4^+(r)}.
\end{aligned}
\end{equation*}

\end{lemma}

The elements of $v_i$ are explicit algebraic expressions in terms of $a_i$ that are defined by \eqref{eq:definition delta_i}. Normalizing the vectors and using expansions of $a_i$ \eqref{eq: series for a} we obtain the following expansion for the matrix $Q$

\begin{lemma}
\label{lem: expansion for Q}
The orthogonal matrix $Q(r)$ of normalized eigenvectors of ${\bf \Delta}(r)$ is
\begin{equation}
\label{eq: expansion for Q}
\begin{aligned}
Q=&\begin{pmatrix}
 0 & 0 & -\dfrac{1}{2} & \dfrac{1}{2} & 0 & \dfrac{1}{\sqrt{2}}\\
-\dfrac{1}{\sqrt{2}} & \dfrac{1}{\sqrt{2}} & 0 & 0 & 0 & 0\\
 0 & 0 & -\dfrac{1}{2} & -\dfrac{1}{2}  & \dfrac{1}{\sqrt{2}}& 0\\
 0 & 0 & \dfrac{1}{2} & -\dfrac{1}{2} & 0 & \dfrac{1}{\sqrt{2}}\\
\dfrac{1}{\sqrt{2}} & \dfrac{1}{\sqrt{2}} & 0 & 0 & 0 & 0\\
0 & 0 & \dfrac{1}{2} & \dfrac{1}{2}  & \dfrac{1}{\sqrt{2}}& 0
\end{pmatrix}
\\
+r^2&
\begin{pmatrix}
0 & 0 & -\dfrac{1}{48} & -\dfrac{1}{48} & -\dfrac{1}{2^5 3 \sqrt{2}} & 0 \\
0 & 0 & 0 & 0 & 0 & 0 \vphantom{\dfrac{1}{2}}\\
0 & 0 & \dfrac{1}{48} & -\dfrac{1}{48} & 0 & -\dfrac{1}{2^5 3 \sqrt{2}}  \\
0 & 0 & \dfrac{1}{48} & \dfrac{1}{48} & -\dfrac{1}{2^5 3 \sqrt{2}} & 0 \\
0 & 0 & 0 & 0 & 0 & 0  \vphantom{\dfrac{1}{2}} \\
0 & 0 & -\dfrac{1}{48} & \dfrac{1}{48} & 0 & \dfrac{1}{2^5 3 \sqrt{2}}  \\
\end{pmatrix}+O(r^4)
\end{aligned}
\end{equation}
\end{lemma}

\section{Expansions of covariance matrices}

\subsection{
\texorpdfstring{Covariance matrix of $(\nabla \Psi(z), \nabla^2 \Psi(z))$}
{Covariance matrix of (Psi'(z),Psi''(z))}
}

\label{cov-exp}

In this section we compute the covariance matrix $\Sigma$ of the $5$-dimensional centred Gaussian vector which combines the gradient and the elements of the Hessian evaluated at $z$.
By the translation invariance of $\Psi$, $\Sigma$ does not depend on the point $z \in \mathbb{R}^2$. It is convenient to write the covariance matrix in blocks
$$\Sigma=\left(\begin{array}{ccc}
A  & B\\
 B^t& C
\end{array} \right),$$
where
\begin{align*}
A=\mathbb{E}[ \nabla \Psi(z)^t \cdot \nabla \Psi(z)], \hspace{1cm} B=\mathbb{E}[ \nabla \Psi(z)^t \cdot  \nabla^2 \Psi(z)], \hspace{1cm}  C=\mathbb{E}[ \nabla^2 \Psi(z)^t \cdot  \nabla^2 \Psi(z)].
\end{align*}
It is a standard fact that covariances of the derivative are given by derivatives of the covariance kernel.

The computations of $A$, $ B$ and $C$ are quite lengthy, but they do not require sophisticated arguments other than iterative differentiation of Bessel functions.  For example to compute $A$ we first have to compute expressions
\begin{equation*}
 \mathbb{E}[ \partial_{z_i} \Psi(z) \; \partial_{w_j} \Psi(w)   ] =
 \frac{\partial^2}{ \partial_{z_i}  \partial_{w_j}  }  \psi(z-w) = \frac{\partial^2}{ \partial_{z_i}  \partial_{w_j}  }  J_0(|z-w|)
\end{equation*}
where $z=(z_1,z_2)$ and $w=(w_1,w_2)$ are two points in $\R^2$. The elements of $A$ are obtained by passing to the limit $w\to z$.

To give an example of such computation we give details of the computation of $(A)_{1,1}$.
For this we first use the chain rule to obtain
\begin{align*}
 \frac{\partial^2}{ \partial_{z_1}  \partial_{w_1}  }  J_0(|z-w|) &=J''_0(|z-w|)\, \left( \frac{\partial}{ \partial_{z_1}   }  |z-w|\right) \;  \left(\frac{\partial}{ \partial_{w_1}  }  |z-w|\right) +J'_0( |z-w|)  \, \frac{\partial^2}{ \partial_{z_1}  \partial_{w_1}  }  |z-w|.
 \end{align*}
The zeroth Bessel function could be defined by power series
$$
J_0(x)=\sum_{n=0}^\infty \frac{(-1)^n}{(n!)^2}\br{\frac{x}{2}}^{2n}.
$$
From this expansion we can immediately get the expansions for $J_0'$ and $J_0''$ and show that
$$
\lim_{w\to z}  \frac{\partial^2}{ \partial_{z_1}  \partial_{w_1}  }  J_0(|z-w|)=\frac{1}{2}.
$$
In the same way we compute the other entries of $A$ and obtain
\begin{align} \label{MA}
A= \left(\begin{array}{cc}
\frac{1}{2} &0 \\
0&\frac{1}{2}
\end{array} \right).
\end{align}
Since the first and second order derivatives of any stationary field are independent at every fixed point $z \in \mathbb{R}^2$, we immediately have $B=0$.
With analogous calculations, but using higher order derivatives of $J_0$ we compute the entries of $C$ and find  that
\begin{align} \label{MC}
C= \left( \begin{array}{ccc}
\frac{3 }{8} &0& \frac{1}{8} \\
0&\frac{1}{8}&0\\
\frac{1}{8}&0&\frac{3}{8}
\end{array}\right).
\end{align}

\subsection{\texorpdfstring{Covariance matrix of $(\nabla \Psi(z), \nabla \Psi(w), \nabla^2 \Psi(z), \nabla^2 \Psi(w))$}{Covariance matrix of (Psi'(z), Psi'(w), Psi''(z), Psi''(w))}} \label{ucn}

We compute the covariance matrix ${\bf \Sigma}(z,w)$ for the $10$-dimensional Gaussian random vector
which combines the gradient and the elements of the Hessian evaluated at $z, w$:
$$(\nabla \Psi(z), \nabla \Psi(w), \nabla^2 \Psi(z), \nabla^2 \Psi(w)),$$
only for the case $z=(0,0)$ and $w=(0,r)$.
As explained in Section \ref{subsec: proof of theorem} this is sufficient in order to evaluate $K_{2}(r)$ for all relevant $r$, thanks to the isotropic property of $\Psi$. It is convenient to write the matrix ${\bf \Sigma}(z,w)$ in block form, i.e.,
$$
{\bf \Sigma}(z,w)={\bf \Sigma}(r)=\left(\begin{array}{ccc}
{\bf A}(z,w)&{\bf B}(z,w)\\
{\bf B}^t(z,w)&{\bf C}(z,w)
\end{array} \right).
$$
The matrix ${\bf A}$ also has a natural block structure
\begin{align*}
&\left. {\bf A}(z,w)\right|_{{z=(0,0),}{w=(0,r)}}={\bf A}(r)= \left(\begin{array}{cc}
A & A(r) \\
A(r)&A
\end{array} \right),
\end{align*}
where $A$ is the same as in \eqref{MA}, and $A(r)$ turns out to be a diagonal matrix, we denote its diagonal elements by $\alpha_i(r)$
\begin{align*}
&A(r)= \left(\begin{array}{cc}
\alpha_1(r) & 0 \\
0&\alpha_2(r)
\end{array} \right)
\end{align*}
The diagonal elements $\alpha_i$ are found by differentiating the covariance kernel of $\Psi$:
\begin{equation}
\begin{aligned}
\alpha_1(r)&=\left. \frac{\partial^2}{\partial_{z_1} \partial_{w_1}} J_0( |z-w|) \right|_{z=(0,0), w=(0,r)}=-J'_0( r)  \frac{1}{r},\\
\alpha_2(r)&=\left. \frac{\partial^2}{\partial_{z_2} \partial_{w_2}} J_0( |z-w|) \right|_{z=(0,0), w=(0,r)}=-J''_0( r) .
\end{aligned}
\end{equation}

Again, using the block structure of $\bf A$ we can write its determinant as
$$
\text{det}({\bf A}(r))=
\br{\alpha_1^2(r) - \frac{1}{4}  }
\br{ \alpha_2^2(r) - \frac{1}{4} }.
$$
From the Taylor series for $J_0$ one immediately gets
\begin{equation*}
\begin{aligned}
\alpha_1(r)&=-J'_0(r) \frac{1}{r}=\frac{1}{2}- \frac{1}{2^4} r^2+O(r^4),
\\
 \alpha_2(r)&=-J''(r)=\frac{1}{2}- \frac{3}{2^4}r^2+O(r^4)
\end{aligned}
\end{equation*}
so that
\begin{align} \label{14:25}
\det(A)=\frac{3r^4}{2^8}+O(r^6).
\end{align}
With analogous calculations we derive also the entries of the matrices ${\bf B}$ and ${\bf C}$: we have
\begin{align*}
\left. {\bf B}(z,w)\right|_{{z=(0,0),}{w=(0,r)}}={\bf B}(r)= \left( \begin{array}{cc} 0 &B(r) \\ - B(r)& 0 \end{array} \right),
\end{align*}
where
\begin{align*}
B(r)=\left( \begin{array}{ccc}
0&\beta_{1}(r) &0 \\
\beta_{1}(r)&0&\beta_{2}(r)
\end{array}\right)
\end{align*}
and, in the same way as before, we obtain explicit formulas in terms of $J_0$ and expansions at $r=0$
\begin{align*}
\beta_1(r)& =-J''_0( r) \frac{1}{r} +J'_0( r) \frac{1}{r^2}=-\frac{r}{8}+\frac{r^3}{96}+O(r^5)\\
  \beta_2(r)& = - J'''_0( r)=-\frac{3r}{8}+\frac{5r^3}{96}+O(r^5).
\end{align*}
In the same way
\begin{align*}
\left. {\bf C}(z,w)\right|_{{z=(0,0),}{w=(0,r)}}=\left( \begin{array}{cc}C & C(r) \\ C(r) &C \end{array} \right),
\end{align*}
where $C$ defined in \eqref{MC} and
\begin{align*}
C(r)= \left( \begin{array}{ccc}
\gamma_{1}(r)&0&\gamma_{2}(r) \\
0 &\gamma_{2}(r)&0 \\
\gamma_{2}(r)&0&\gamma_{3}(r)
\end{array}\right),
\end{align*}
with
\begin{align*}
\gamma_1(r)&= J''_0( r)\frac{3 }{r^2} - J'_0( r)\frac{3 }{r^3}=\frac{3}{8}  -\frac{r^2}{32} + O(r^4),
\\
\gamma_2(r)&=  \frac{J'''_0( r)}{r} - J''_0( r) \frac{2 }{r^2} +  J'_0( r)\frac{2 }{r^3}=\frac{1}{8} - \frac{r^2}{32} + O(r^4),
\\
 \gamma_3(r)& = J''''_0( r)=\frac{3}{8} - \frac{5 r^2}{32}+O(r^4).
\end{align*}

\subsection{Conditional covariance matrix} \label{matrixdelta}

As explained before, the covariance matrix of the conditional vector
$$
(\nabla^2 \Psi(z), \nabla^2 \Psi(w) | \nabla \Psi(z)= \nabla \Psi(w)=0)
$$
is given by
$$
{\bf \Delta}(r)={\bf C}(r)-{\bf B}(r)^t {\bf A}(r)^{-1} {\bf B}(r)=\left( \begin{array}{cc}  \Delta_1(r) &  \Delta_2(r) \\ \Delta_2(r) & \Delta_1(r)
\end{array} \right)
$$
where $\Delta_i$ are defined in \eqref{eq:definition delta_i}.

Since we already have explicit formulas for elements of $\bf A$, $\bf B$, and $\bf C$
we can obtain the following explicit formulas and expansions for $a_i$ that define $\bf \Delta$:

\begin{equation}
\label{eq: series for a}
\begin{aligned}
a_1(r)&= - \frac{2 \beta_1^2(r)}{1-4 \alpha_2^2(r)}+\frac{1}{2^3 3}=- \frac{13}{2^7 3^3} r^2- \frac{151}{2^{11} 3^5} r^4- \frac{1531}{2^{15} 3^7}r^6+O(r^8),\\
a_2(r)&= -\frac{2 \beta_1^2(r)}{1-4 \alpha_1^2(r)}+\frac{1}{2^3} = \frac{1}{2^7} r^2 + \frac{1}{2^{11} 3^2} r^4 - \frac{23}{2^{15} 3^3 5} r^6+O(r^8),\\
a_3(r)&=- \frac{2 \beta_2^2(r)}{1-4 \alpha_2^2(r)}+ \frac{3}{2^3} =\frac{1}{2^7} r^2 + \frac{41}{2^{11} 3^3}r^4+ \frac{2617}{2^{15} 3^5 5 }r^6 +O(r^8),\\
a_4(r)&= - \frac{2 \beta_1(r) \beta_2(r)}{1-4 \alpha_2^2(r)}+ \frac{1}{2^3}= - \frac{5}{2^7 3^2} r^2 - \frac{23}{2^{11} 3^4} r^4 + \frac{521}{2^{15} 3^6 5} r^6 + O(r^8),\\
a_5(r)&=\gamma_1(r)- \frac{4 \alpha_2(r) \beta_1^2(r) }{1-4 \alpha_2^2(r)}- \frac{1}{3} = - \frac{67}{2^7 3^3} r^2 +\frac{7 \cdot 71}{2^{11} 3^5}r^4 +\frac{13 \cdot 271}{2^{15} 3^7 5}r^6 +O(r^8),\\
a_6(r)&= \gamma_2(r) - \frac{4 \alpha_1(r) \beta_1^2(r)}{1-4 \alpha_1^2(r)}= -\frac{1}{2^7} r^2 + \frac{1}{2^{11} 3^2} r^4 +\frac{19}{2^{15} 3^3 5}r^6 +O(r^8), \\
a_7(r)&= \gamma_3(r)- \frac{4 \alpha_2(r) \beta_2^2(r)}{1-4 \alpha_2^2(r)} = -\frac{1}{2^7}r^2 -\frac{31}{2^{11} 3^3} r^4 -\frac{2621}{2^{15} 3^5 5} r^6 +O(r^8),\\
a_8(r)&= \gamma_2(r) - \frac{4 \alpha_2(r) \beta_1(r) \beta_2(r)}{1-4 \alpha_2^2(r)}= \frac{13}{2^7 3^2}r^2 -\frac{23}{2^{11} 3^4}r^4 + \frac{7}{2^{15} 3^6}r^6 +O(r^8).
\end{aligned}
\end{equation}

\bibliography{RPW}{}

\begin{thebibliography}{10}

\bibitem{Adler-Taylor}
R.J. Adler and J.E. Taylor.
\newblock {\em Random fields and geometry}.
\newblock Springer Monographs in Mathematics. Springer, New York, 2007.

\bibitem{AzWsch}
J.~Aza{\"{\i}}s and M.~Wschebor.
\newblock {\em Level sets and extrema of random processes and fields}.
\newblock John Wiley \& Sons, Inc., Hoboken, NJ, 2009.

\bibitem{Berry77}
M.V. Berry.
\newblock Regular and irregular semiclassical wavefunctions.
\newblock {\em Journal of Physics A: Mathematical and General}, 10(12):2083,
  1977.

\bibitem{CaMaWi}
Valentina Cammarota, Domenico Marinucci, and Igor Wigman.
\newblock On the distribution of the critical values of random spherical
  harmonics.
\newblock {\em The Journal of Geometric Analysis}, 26(4):3252--3324, 2016.

\bibitem{CaWi}
Valentina Cammarota and Igor Wigman.
\newblock Fluctuations of the total number of critical points of random
  spherical harmonics.
\newblock {\em arXiv preprint arXiv:1510.00339}, 2015.

\bibitem{CaHa16}
Yaiza Canzani and Boris Hanin.
\newblock C-infinity scaling asymptotics for the spectral function of the
  laplacian.
\newblock {\em arXiv preprint arXiv:1602.00730}, 2016.

\bibitem{Estrade}
A.~Estrade and J.~Fournier.
\newblock Boundedness of level lines for two-dimensional random fields.
\newblock {\em Statistics \& Probability Letters}, 118:94--99, 2016.

\bibitem{G-W}
D.~Gayet and J.-Y. Welschinger.
\newblock Lower estimates for the expected {B}etti numbers of random real
  hypersurfaces.
\newblock {\em J. Lond. Math. Soc. (2)}, 90(1):105--120, 2014.

\bibitem{Hormander}
L.~H\"{o}rmander.
\newblock The spectral function of an elliptic operator.
\newblock {\em Acta Math.}, 121:193--218, 1968.

\bibitem{Lax}
Peter~D. Lax.
\newblock Asymptotic solutions of oscillatory initial value problems.
\newblock {\em Duke Math. J}, 24:627--646, 1957.

\bibitem{NS09}
F.~Nazarov and M.~Sodin.
\newblock On the number of nodal domains of random spherical harmonics.
\newblock {\em Amer. J. Math.}, 131(5):1337--1357, 2009.

\bibitem{NS15}
F.~Nazarov and M.~Sodin.
\newblock Asymptotic laws for the spatial distribution and the number of
  connected components of zero sets of {G}aussian random functions.
\newblock {\em J. Math. Phys. Anal. Geo.}, 12(3):205--278, 2016.

\bibitem{So12}
Mikhail Sodin.
\newblock Lectures on random nodal portraits.
\newblock {\em Lecture notes for a mini-course given at the St. Petersburg
  Summer School in Probability and Statistical Physics (June, 2012) Available
  at: http://www.math.tau.ac.il/sodin/SPB-Lecture-Notes. pdf}, 2012.

\bibitem{Zelditch09}
S.~Zelditch.
\newblock Real and complex zeros of {R}iemannian random waves.
\newblock In {\em Spectral analysis in geometry and number theory}, volume 484
  of {\em Contemp. Math.}, pages 321--342. Amer. Math. Soc., Providence, RI,
  2009.

\end{thebibliography}
\bibliographystyle{plain}

\end{document}